\documentclass[a4paper,UKenglish,numberwithinsect]{lipics-v2018}

\usepackage{mathpartir}
\usepackage{mathtools}
\usepackage{microtype}
\usepackage{stmaryrd}
\usepackage[all,pdf]{xy}
\CompileMatrices % xy-pic compiles diagrams
\newcommand{\pullbackcorner}[1][dr]{\save*!/#1-1.4pc/#1:(-1,1)@^{|-}\restore}

\theoremstyle{plain}

\newtheorem{rem}[theorem]{Remark}
\newtheorem{propo}[theorem]{Proposition}

\newcommand{\ACT}{`}
\newcommand{\AX}{\mathsf{ax}}
\newcommand{\bij}{\cong}
\newcommand{\BOOL}{\mathbb{B}}
\newcommand{\cat}{\mathbf{C}}
\newcommand{\C}{\mathsf{C}}
\newcommand{\CODE}{\mathsf{code}}
\newcommand{\CODEEL}{\mathsf{codeEl}}
\newcommand{\COF}{\mathsf{cof}}
\newcommand{\COFAND}{\mathsf{cofAnd}}
\newcommand{\COFEMPTY}{\mathsf{cof}_{\bot}}
\newcommand{\COFEXT}{\mathsf{cofExt}}
\newcommand{\COFI}{\mathsf{cofI}}
\newcommand{\COFO}{\mathsf{cofO}}
\newcommand{\COFPROP}{\mathsf{cofisProp}}
\newcommand{\COFOR}{\mathsf{cofOr}}
\newcommand{\comp}{\circ}
\newcommand{\COMP}{\mathsf{CCHM}}
\newcommand{\CONG}{\mathsf{cong}}
\newcommand{\COP}{{\scriptstyle\sqrt{\mathstrut}}}
\newcommand{\Cset}{\widehat{\Box}}
\newcommand{\Cubes}{\Box}

\newcommand{\Ec}{\diamond}
\newcommand{\EL}{\mathsf{El}}
\newcommand{\ELCODE}{\mathsf{Elcode}}
\newcommand{\ELT}{\mathsf{Elt}}
\newcommand{\EMPTY}{\bot}
\newcommand{\EMPTYELIM}{\bot\mathsf{elim}}
\newcommand{\EQ}{\equiv}
\newcommand{\EXT}{\mathbin{\nnearrow}}
\newcommand{\FIB}{\mathsf{Fib}}
\newcommand{\FST}{\mathsf{fst}}
\newcommand{\FUN}{\shortrightarrow}
\newcommand{\FUNEXT}{\mathsf{funext}}
\newcommand{\HOLE}{\text{\texttt{\_}}}
\newcommand{\I}{\mathtt{I}}

\newcommand{\ID}{\mathsf{id}}
\newcommand{\INONTRIV}{\O{\not\EQ}\I}
\newcommand{\Int}{\mathrm{I}}
\newcommand{\INT}{\mathbb{I}}
\newcommand{\INTUNIV}{\mathsf{IntUniv}}
\newcommand{\ISFIB}{\mathsf{isFib}}

\renewcommand{\L}{\mathsf{L}}

\newcommand{\MIN}{\sqcap}
\newcommand{\mono}{\rightarrowtail}
\newcommand{\morphism}{\rightarrow}
\newcommand{\Nat}{\mathbb{N}}
\renewcommand{\O}{\mathtt{O}}
\newcommand{\op}{\mathrm{op}}
\newcommand{\OR}{\vee}
\newcommand{\PAIR}[2]{(#1\mathbin{,}#2)}
\renewcommand{\PATH}{\wp}
\newcommand{\PR}{\mathsf{pr}}
\newcommand{\PROPCOF}{\mathsf{isPropcof}}
\newcommand{\R}{\mathsf{R}}
\newcommand{\REFL}{\mathsf{refl}}
\newcommand{\REV}{\mathsf{rev}}
\newcommand{\Set}{\mathcal{S}\!\mathit{et}}
\newcommand{\SET}{\mathsf{Set}}
\newcommand{\SIGMA}[2]{\Sigma\,{#1:#2}\mathbin{,}}
\newcommand{\SND}{\mathsf{snd}}
\newcommand{\STRAX}{\mathsf{strax}}
\newcommand{\U}{\mathsf{U}}
\newcommand{\UIP}{\mathsf{uip}}
\newcommand{\UNIT}{\top}
\newcommand{\UNITVAL}{\mathsf{tt}}

\newcommand{\y}{\mathrm{y}}

\title{Internal Universes in Models of Homotopy Type Theory}

\author{Daniel R. Licata}{Wesleyan University Dept.~Mathematics \&
  Computer Science, Middletown,
  USA}{}{https://orcid.org/0000-0003-0697-7405}{Supported by United
  States Air Force Research Laboratory, agreement numbers
  FA-95501210370 and FA-95501510053.}

\author{Ian Orton}{University of Cambridge Dept.~Computer Science \&
  Technology, Cambridge,
  UK}{}{https://orcid.org/0000-0002-9924-0623}{Supported by a UK EPSRC
  PhD studentship, funded by grants EP/L504920/1, EP/M506485/1.}

\author{Andrew M. Pitts}{University of Cambridge Dept.~Computer
  Science \& Technology, Cambridge,
  UK}{}{https://orcid.org/0000-0001-7775-3471}{}

\author{Bas Spitters}{Aarhus University Dept.~Computer Science,
  Aarhus, DK}{}{https://orcid.org/0000-0002-2802-0973}{Supported by
  the Guarded Homotopy Type Theory project, funded by the Villum
  Foundation, project number 12386.}

\authorrunning{D. R. Licata, I. Orton, A.~M. Pitts and B. Spitters} 

\Copyright{Daniel R. Licata, Ian Orton, Andrew M. Pitts and Bas
  Spitters} 

\subjclass{\ccsdesc[500]{Theory of computation~Type theory}}

\keywords{cubical sets, dependent type theory, homotopy type theory,
  internal language, modalities, univalent foundations, universes}

\relatedversion{https://arxiv.org/abs/1801.07664}

\supplement{\url{https://doi.org/10.17863/CAM.22369} (Agda-flat source code)} 

\acknowledgements{We benefited from discussions with Steve Awodey, Thierry Coquand,
   Christian Sattler, and Mike Shulman. Andrea Vezzosi's Agda-flat was
  invaluable for developing and checking some of the results in this
  paper. The third author thanks Aarhus University Department of
  Computer Science for hosting him while some of the work was carried
  out.}

\EventEditors{H\'{e}l\`{e}ne Kirchner}
\EventNoEds{1}
\EventLongTitle{3rd International Conference on Formal Structures for
  Computation and Deduction (FSCD 2018)} 
\EventShortTitle{FSCD 2018}
\EventAcronym{FSCD}
\EventYear{2018}
\EventDate{July 9--12, 2018}
\EventLocation{Oxford, UK}
\EventLogo{}
\SeriesVolume{108}
\ArticleNo{22} % “New number” (=<article-no>) goes here!
\nolinenumbers %uncomment to disable line numbering

\hideLIPIcs %uncomment to remove references to LIPIcs series (logo,
\begin{document}

\maketitle

\begin{abstract}
  We begin by recalling the essentially global character of universes
  in various models of homotopy type theory, which prevents a
  straightforward axiomatization of their properties using the
  internal language of the presheaf toposes from which these model are
  constructed. We get around this problem by extending the internal
  language with a modal operator for expressing properties of global
  elements. In this setting we show how to construct a universe that
  classifies the Cohen-Coquand-Huber-M\"ortberg (CCHM) notion of
  fibration from their cubical sets model, starting from the
  assumption that the interval is tiny---a property that the interval
  in cubical sets does indeed have. This leads to an elementary
  axiomatization of that and related models of homotopy type theory
  within what we call \emph{crisp type theory}.
\end{abstract}

%%%%%%%%%%%%%%%%%%%%%%%%%%%%%%%%%%%%%%%%%%%%%%%%%%%%%%%%%%%%%%%%%%%%%%
\section{Introduction}
\label{sec:int}

Voevodsky's univalence axiom in Homotopy Type Theory
(HoTT)~\cite{HoTT} is motivated by the fact that constructions on
structured types should be invariant under isomorphism. From a
programming point of view, such constructions can be seen as
type-generic programs. For example, if $G$ and $H$ are isomorphic
groups, then for any construction $C$ on groups, an instance $C(G)$
can be \emph{transported} to $C(H)$ by lifting this isomorphism using
a type-generic program corresponding to $C$. As things stand, there is
no single definition of the semantics of such generic programs;
instead there are several variations on the theme of giving a
computational interpretation to the new primitives of HoTT (univalence
and higher inductive types) via different constructive
models~\cite{CoquandT:modttc,CoquandT:cubttc,AngiuliC:comhtt,LicataDR:carctt},
the pros and cons of which are still being explored.

As we show in this paper, that exploration benefits from being carried
out in a type-theoretic language.  This is different from developing
the consequences of HoTT itself using a type-theoretic language, such
as intensional Martin-L\"of type theory with axioms for univalence and
higher inductive types, as used in \cite{HoTT}. There \emph{all} types
have higher-dimensional structure, or ``are fibrant'' as one says, via
the structure of the iterated identity types associated with
them. Contrastingly, when using type theory to describe models of
HoTT, being fibrant is an explicit structure external to a type; and
that structure can itself be classified by a type, so that users of
the type theory can \emph{prove} that a type is fibrant by inhabiting
a certain other type.  As an example, consider the cubical sets model
of type theory introduced by Cohen, Coquand, Huber and
M\"ortberg~(CCHM)~\cite{CoquandT:cubttc}. This model uses a presheaf
topos on a particular category of cubes that we denote by $\Cubes$,
generated by an interval object $\INT$, maps out of which represent
paths.  The corresponding presheaf topos $\hat{\Cubes}$ has an
associated category with families (CwF)~\cite{DybjerP:intt} structure
that gives a model of Extensional Martin-L\"of Type
Theory~\cite{Martin-LoefP:inttt} in a standard
way~\cite{HofmannM:synsdt}.  While not all types in this presheaf
topos have a fibration structure in the CCHM sense, working within
constructive set theory, CCHM show how to make a new CwF of fibrant
types out of this presheaf CwF, one which is a model of Intensional
Martin-L\"of Type Theory with univalent universes and (some) higher
inductive types~\cite{HoTT}. Their model construction is rather subtle
and complicated.  Coquand noticed that the CCHM version of Kan
fibration could be more simply described in terms of partial elements
in the \emph{internal language} of the topos.  Some of us took up and
expanded upon that suggestion in~\cite{PittsAM:aximct} and
\cite[Section~4]{SpittersB:guactt}. Using Extensional Martin-L\"of Type
Theory with an impredicative universe of propositions (one candidate
for the internal language of toposes), those works identify some
relatively simple axioms for an interval and a collection of
Kan-filling shapes (\emph{cofibrant} propositions) that are sufficient
to define a CwF of CCHM fibrations and prove most of its properties as
a model of univalent foundations, for example, that $\Pi$, $\Sigma$,
path and other types are fibrant.  These internal language
constructions can be used as an intermediate point in constructing a
concrete model in cubical sets: the type theory of HoTT~\cite{HoTT}
can be translated into the internal language of the topos, which has a
semantics in the topos itself in a standard way.  The advantages of
this indirection are two-fold. First, the definition and properties of
the notion of fibration (both the CCHM notion~\cite{CoquandT:cubttc}
and other related ones~\cite{LicataDR:carctt,ShulmanM:typtsc}) are
simpler when expressed in the internal language; and secondly, so long
as the axioms are not too constraining, it opens up the possibility of
finding new models of HoTT. Indeed, since our axioms do not rely on
the infinitary aspects of Grothendieck toposes (such as having
infinite colimits), it is possible to consider models of them in
elementary toposes, such as Hyland's effective
topos~\cite{FruminD:homtmf,UemuraT:cubaip}.

From another point of view, the internal language of the presheaf
topos can itself be viewed as a two-level type
theory~\cite{AltenkirchT:exthtt,VoevodskyV:simtst} with fibrant and
non-fibrant types, where being fibrant is classified by a type, and
the constructions are a library of fibrancy instances for all of the
usual types of type theory.  Directed type
theory~\cite{ShulmanM:typtsc} has a very similar story: it adds a
directed interval type and a logic of partial elements to homotopy
type theory, and using them defines some new notions of
higher-dimensional structure, including co- and contravariant
fibrations.

However, the existing work describing models using an internal
language \cite{PittsAM:aximct, SpittersB:guactt, LicataDR:carctt}
does not encompass \emph{universes} of fibrant types.  The lack of
universes is a glaring omission for making models of HoTT, due to both
their importance and the difficulty of defining them correctly.
Moreover, it is an impediment to using internal language presentations
of cubical type theory as a two-level type theory. For example, most
constructions on higher inductive types, like calculating their
homotopy groups, require a fibrant universe of fibrant types; and
adding universes to directed type theory would have analogous
applications.  Finally, packaging the fibrant types together into a
universe restores much of the convenience of working in a language
where all types are fibrant: instead of passing around separate
fibrancy proofs, one knows that a type is fibrant by virtue of the
universe to which it belongs.

In this paper, we address this issue by studying universes of fibrant
types expressed in internal languages for models of cubical type
theories.  CCHM \cite{CoquandT:cubttc} define a universe concretely
using a version of the Hofmann-Streicher universe construction in
presheaf toposes~\cite{HofmannM:lifgu}. This gives a classifier for
their notion of fibration---the universe is equipped with a CCHM
fibration that gives rise to every fibration (with small fibres) by
re-indexing along a function into the universe.  In this way one gets
a model of a Tarski-style universe closed under whatever type-forming
operations are supported by CCHM fibrations.  Thus, there is an
appropriate semantic target for a universe of fibrant types, but
neither \cite{PittsAM:aximct}, nor \cite{SpittersB:guactt} gave a
version of such a universe expressed in the internal language.  This
is for a good reason: \cite[Remark~7.5]{PittsAM:aximct-jv} points out
that there can be no \emph{internal} universe of types equipped with a
CCHM fibration that weakly classifies fibrations. We recall in detail
why this is the case in Section~\ref{sec:nogti}, but the essence is
that na\"ive axioms for a weak classifier for fibrations imply that a
family of types, each member of which is fibrant, has to form a
fibrant family; but this is not true for many notions of fibration,
such as the CCHM one.

To fix this issue, in Section~\ref{sec:loctt} we enrich the internal
language to a \emph{modal} type theory with two context
zones~\cite{PfenningF:judrml,DePaivaV:fibmtt,ShulmanM:brofpt},
inspired in particular by the fact that cubical sets are a model of
Shulman's spatial type theory.  In a judgement
$\Delta \mid \Gamma \vdash a : A$ of this modal type theory, the
context $\Gamma$ represents the usual local elements of types in the
topos, while the new context $\Delta$ represents global ones.  The dual
context structure is that of an S4 necessity modality in modal logic,
because a global element determines a local one, but global elements
cannot refer to local elements.  We use Shulman's term ``crisp'' for
variables from $\Delta$, and call the type theory \emph{crisp type
  theory}, because we do not in fact use any of the modal type
operators of his spatial type theory, but just $\Pi$-types
whose domains are crisp.  Using these crisp $\Pi$-types, we give
axioms that specify a universe that classifies global
fibrations---the modal structure forbids the internal substitutions
that led to inconsistency.

One approach to validating these universe axioms would be to check
them directly in a cubical set model; but we can in fact do more work
using crisp type theory as the internal language and reduce the
universe axioms to a structure that is simpler to check in
models. Specifically, in Theorem~\ref{thm:universe}, we construct such
a universe from the assumption that the interval $\INT$ is
\emph{tiny}, which by definition means that its exponential functor
$\INT\FUN\HOLE$ has a right adjoint (a global one, not an internal
one---this is another example where crisp type theory is needed to
express this distinction). The ubiquity of right adjoints to
exponential functors was first pointed out by
Lawvere~\cite{LawvereFW:towdst} in the context of synthetic
differential geometry. Awodey pointed out their occurrence in
interval-based models of type theory in his work on various cube
categories~\cite{AwodeyS:notcmt}.  As far as we know, it was Sattler
who first suggested their relevance to constructing universes in such
models (see \cite[Remark~8.3]{SattlerC:equepm}).  It is indeed the
case that the interval object in the topos of cubical sets is
tiny. Some ingenuity is needed to use the right adjoint to
$\INT\FUN\HOLE$ to construct a universe with a fibration that gives
rise to every other one up to equality, rather than just up to
isomorphism; we employ a technique of
Voevodsky~\cite{VoevodskyV:csyduc} to do so.

Finally, we describe briefly some applications in
Section~\ref{sec:modhtt}.  First, our universe construction based on a
tiny interval is the missing piece that allows a completely internal
development of a model of univalent foundations based upon the CCHM
notion of fibration, albeit internal to crisp type theory rather than
ordinary type theory.  Secondly, we describe a preliminary result
showing that our axioms for universes are suitable for building type
theories with hierarchies of universes, each with a different notion
of fibration.

The constructions and proofs in this paper have been formalized in
Agda-flat~\cite{Agda-flat}, an appealingly simple extension of
Agda~\cite{Agda} that implements crisp type theory; see
\url{https://doi.org/10.17863/CAM.22369}. Agda-flat
was provided to us by Vezzosi as a by-product of his work on modal
type theory and parametricity~\cite{VezzosiA:parqdt}.

%%%%%%%%%%%%%%%%%%%%%%%%%%%%%%%%%%%%%%%%%%%%%%%%%%%%%%%%%%%%%%%%%%%%%%
\section{Internal description of fibrations}
\label{sec:intdcf}

We begin by recalling from \cite{PittsAM:aximct-jv,SpittersB:guactt}
the internal description of fibrations in presheaf models, using CCHM
fibrations~\cite[Definition~13]{CoquandT:cubttc} as an example. Rather
than using Extensional Martin-L\"of Type Theory with an impredicative
universe of propositions as in
\cite{PittsAM:aximct-jv,SpittersB:guactt}, here we use an intensional
and predicative version, therefore keeping within a type theory with
decidable judgements.\footnote{Albeit at the expense of some
  calculations with universe levels; Coq's universe polymorphism would
  probably deal with this aspect automatically.}  Our type theory of
choice is the one implemented by Agda~\cite{Agda}, whose assistance we
have found invaluable for developing and checking the definitions.
Adopting Agda-style syntax, dependent function types are written
$(x : A) \FUN B\,x$, or $\{x : A\} \FUN B\,x$ if the argument to the
function is implicit; non-dependent function types are written
$(\HOLE:A)\FUN B$, or just $A\FUN B$. There is a non-cumulative
hierarchy of Russell-style~\cite{LuoZ:notutt} universe types
$\SET=\SET_0 : \SET_1 : \SET_2 : \SET_3 \ldots$ Among Agda's inductive
types we need identity types
$\HOLE\EQ\HOLE : \{A : \SET_n\}\FUN A \FUN A \FUN \SET_n$, which form
the inductively defined family of types with a single constructor
$\REFL : \{A : \SET_n\}\{x:A\} \FUN {x \EQ x}$; and we need the empty
inductive type $\EMPTY:\SET$, which has no constructors.  Among Agda's
record types (inductive types with a single constructor for which
$\eta$-expansion holds definitionally) we need the unit type
$\top:\SET$ with constructor $\UNITVAL:\UNIT$; and dependent products
($\Sigma$-types), that we write as $\SIGMA{x}{A} B\,x$ and which are
dependent record types with constructor
$\PAIR{\HOLE}{\HOLE}:(x : A)(\HOLE : B\,x) \FUN \SIGMA{x}{A} B\,x$ and
fields (projections) $\FST: (\SIGMA{x}{A} B\,x) \FUN A$ and
$\SND : (z : \SIGMA{x}{A} B\,x) \FUN B(\FST\,z)$.

This type theory can be interpreted in (the category with families of)
any presheaf topos, such as the one defined below, so long as we
assume that the ambient set theory has a countable hierarchy of
Grothendieck universes; in particular, one could use a constructive
ambient set theory such as IZF~\cite{AczelP:reltts} with universes. We
will use the fact that the interpretation of the type theory in
presheaf toposes satisfies \emph{function extensionality} and
\emph{uniqueness of identity proofs}:
\begin{align}
  \label{eq:1}
  &\FUNEXT : \{A : \SET_n\}\{B : A \FUN\SET_m\}\{f\;g : (x : A) \FUN
  B\,x\}((x : A) \FUN f\,x \EQ g\,x) \FUN {f \EQ g} \\
  \label{eq:2} 
  &\UIP : \{A : \SET_n\}\{x\;y : A\}\{p\;q : x \EQ y\} \FUN {p \EQ q} 
\end{align}

\begin{definition}[\textbf{Presheaf topos of de Morgan cubical sets}]
  \label{def:cset}
  Let $\Cubes$ denote the small category with finite products which is
  the Lawvere theory of De~Morgan algebra
  (see~\cite[Chap.~XI]{BalbesR:disl} and
  \cite[Section~2]{SpittersB:cubstt}). Concretely, $\Cubes^\op$ consists
  of the free De~Morgan algebras on $n$ generators, for each
  $n\in\Nat$, and the homomorphisms between them. Thus $\Cubes$
  contains an object $\Int$ that generates the others by taking finite
  products, namely the free De~Morgan algebra on one generator. This
  object is the generic De~Morgan algebra and in particular it has two
  distinct global elements, corresponding to the constants for the
  greatest and least elements.  The \emph{topos of cubical
    sets}~\cite{CoquandT:cubttc}, which we denote by $\Cset$, is the
  category of $\Set$-valued functors on $\Cubes^\op$ and natural
  transformations between them. The Yoneda embedding, written
  $\y:\Cubes\hookrightarrow\Cset$, sends $\Int\in\Cubes$ with its two
  distinct global elements to a representable presheaf $\INT = \y\Int$
  with two distinct global elements.  This interval $\INT$ is used to
  model path types: a path in $A$ from $a_0$ to $a_1$ is any morphism
  $\INT \to A$ that when composed with the distinct global elements
  gives $a_0$ and $a_1$.
\end{definition}
The toposes used in other cubical
models~\cite{CoquandT:modttc,AngiuliC:comhtt,LicataDR:carctt} vary the
choice of algebra from the De~Morgan case used above;
see~\cite{BuchholtzU:varcs}.  To describe all these cubical models
using type theory as an internal language, we postulate the existence
of an \emph{interval} type $\INT$ with two distinct elements, which we
write as $\O$ and $\I$:
\begin{equation}
  \INT : \SET
  \qquad\qquad
  \O :\INT
  \qquad\qquad
  \I : \INT
  \qquad\qquad
  \INONTRIV : (\O \EQ \I)\FUN\EMPTY
  \label{eq:3}\\
\end{equation}
Apart from an interval, the other data needed to define a cubical sets
model of homotopy type theory is a notion of \emph{cofibration}, which
specifies the shapes of filling problems that can be solved in a
dependent type.  For this, CCHM~\cite{CoquandT:cubttc} use a
particular subobject of $\Omega\in \Cset$ (the subobject classifier in
the topos $\Cset$), called the \emph{face lattice}; but other choices
are possible~\cite{PittsAM:aximct-jv}. Here, we avoid the use of the
impredicative universe of propositions $\Omega$ and just assume the
existence of a collection of ``cofibrant'' types in the first universe
$\SET$, including at least the empty type $\EMPTY$ (in
Section~\ref{sec:modhtt}, we will introduce more cofibrations, needed
to model various type constructs):
\begin{equation}
  \label{eq:4}
  \COF : \SET\FUN\SET
  \qquad
  \COFEMPTY: \COF\,\EMPTY
\end{equation}
We call $\varphi:\SET$ \emph{cofibrant} if $\COF\,\varphi$ holds, that
is, if we can supply a term of that type.  To define the
fibrations as a type in the internal language we
use two pieces of notation. First, the \emph{path functor}
associated with the interval $\INT$ is
\begin{align}
  &\PATH : \SET_n\FUN\SET_n
  &&\PATH\ACT : \{A:\SET_n\}\{B:\SET_m\}(f: A \FUN  B) \FUN \PATH\,A
     \FUN\PATH\,B \label{eq:5}\\ 
  &\PATH\,A = \INT \FUN A
  &&\PATH\ACT f\,p\,i = f(p\,i) \notag
\end{align}
Secondly, we define the following \emph{extension} relation
\begin{equation}
   \label{eq:6}
  \HOLE\EXT\HOLE : \{\varphi : \SET\}\{A : \SET_n\}(t : \varphi \FUN
  A)(x : A) \FUN \SET_n
  \qquad
  {t\EXT x} = (u:\varphi)\FUN{t\,u\EQ x}
\end{equation}
Thus $t\EXT x$ is the type of proofs that the partial element
$t:\varphi\FUN A$ extends to the (total) element $x:A$.  We will use
this when $t$ denotes a partial element of $A$ of cofibrant extent,
that is when we have a proof of $\COF\,\varphi$.

\begin{definition}[\textbf{fibrations}] 
  \label{def:cchmf}
  The type $\ISFIB\,A$ of \emph{fibration structures} for a family of
  types $A:\Gamma\FUN\SET_n$ over some type $\Gamma:\SET_m$ consists
  of functions taking any path $p:\PATH\,\Gamma$ in the base type to a
  \emph{composition structure} in $\C(A\comp p)$:
  \begin{equation}
    \label{eq:7}
    \ISFIB : (\Gamma:\SET_m)(A:\Gamma\FUN\SET_n)
    \FUN \SET_{1\sqcup m \sqcup n}
    \qquad
    \ISFIB\,\Gamma\,A = (p :\PATH\,\Gamma)\FUN
      \C(A\comp p)
    \end{equation}
    Here $\C$ is some given function
    $\PATH\,\SET_n \FUN\SET_{1\sqcup n}$ (polymorphic in the universe
    level $n$) which parameterizes the notion of fibration. Then for
    each type $\Gamma$, the type $\FIB_n\,\Gamma$ of \emph{fibrations}
    over it with fibers in $\SET_n$ consists of families equipped with
    a fibration structure
  \begin{equation}
    \label{eq:8}
    \FIB_n  : (\Gamma : \SET_m) \FUN  \SET_{m
      \sqcup(n+1)}
    \qquad
    \FIB_n\,\Gamma = \SIGMA{A}{(\Gamma\FUN\SET_n)}
      \ISFIB\,\Gamma\,A 
    \end{equation}
  and there are \emph{re-indexing} functions, given by
  composition of dependent functions ($\HOLE\comp\HOLE$) 
  \begin{align}
    &\HOLE[\HOLE] : \{\Gamma : \SET_k\}\{\Gamma' : \SET_m\}
      (\Phi : \FIB_n\,\Gamma)(\gamma : \Gamma' \FUN \Gamma) \FUN
      \FIB_n\,\Gamma'\\ 
    &\PAIR{A}{\alpha}[\gamma] =
      \PAIR{A\comp f}{\alpha\comp\PATH\ACT\,f}\notag
  \end{align}
  A \emph{CCHM fibration} is the above notion of fibration for the
  composition structure $\COMP: \PATH\,\SET_n \FUN\SET_{1\sqcup n}$
  from~\cite{CoquandT:cubttc}:
  \begin{equation}
    \label{eq:22}
    \COMP\,P =
    \begin{array}[t]{@{}r}
      (\varphi:\SET)(\HOLE:\COF\,\varphi)(p : (i :\INT)
      \FUN \varphi \FUN P\,i) \FUN 
      (\SIGMA{a_0}{P\,\O} p\,\O\EXT a_0) \FUN \quad\mbox{}\\
      (\SIGMA{a_1}{P\,\I} p\,\I\EXT a_1)   
    \end{array}   
  \end{equation}
  Thus the type $\COMP\,P$ of CCHM composition structures for a path
  of types $P:\PATH\,\SET_n$ consists of functions taking any
  dependently-typed path of partial elements
  $p : (i :\INT) \FUN \varphi \FUN P\,i$ of cofibrant extent to a
  function mapping extensions of the path at one end $p\,\O\EXT a_0$,
  to extensions of it at the other end $p\,\I\EXT a_1$.  When the
  cofibration is $\EMPTY$, this $\ISFIB\,\Gamma\,A$ expands to the
  statement that for all paths $p : \INT \to \Gamma$,
  $A(p\, \O) \to A(p\, \I)$, so that this internal language type says
  that $A$ is equipped with a transport function along paths in
  $\Gamma$.  The use of cofibrant partial elements generalizes
  transport with a notion of path composition, which is used to show
  that path types are fibrant.
\end{definition}

Other notions of fibration follow the above definitions but vary the
definition of $\C: \PATH\,\SET_n \FUN\SET_{1\sqcup n}$; for
example, generalized diagonal Kan composition~\cite{LicataDR:carctt}.
Co/contravariant fibrations in directed type
theory~\cite{ShulmanM:typtsc} also have the form of $\ISFIB$
for some $\C$, but with $\PATH$ being directed paths.
Definition~\ref{def:cchmf} illustrates the advantages of
internal-language presentations; in particular,
uniformity~\cite{CoquandT:cubttc} is automatic.

If $\Gamma$ denotes an object of the cubical sets topos $\Cset$, then
$\FIB_0\,\Gamma$ denotes an object whose global sections correspond to
the elements of the set $\mathrm{FTy}(\Gamma)$ of families over
$\Gamma$ equipped with a composition structure as defined
in~\cite[Definition~13]{CoquandT:cubttc}.  Our goal now is to first
recall that there can be no universe that weakly classifies these CCHM
fibrations in an internal sense, and then move to a modal type theory
where such a universe can be expressed.

%%%%%%%%%%%%%%%%%%%%%%%%%%%%%%%%%%%%%%%%%%%%%%%%%%%%%%%%%%%%%%%%%%%%%%
\section{The "no-go" theorem for internal universes}
\label{sec:nogti}

In this section we recall from \cite[Remark~7.5]{PittsAM:aximct-jv}
why there can be no universe that weakly classifies CCHM fibrations in
an internal sense. Such a weak classifier would be given by the
following data
\begin{equation}
  \label{eq:9}
  \begin{aligned}
    &\U : \SET_2
    && \CODE : \{\Gamma:\SET\}(\Phi : \FIB_0\,\Gamma) \FUN
    \Gamma \FUN \U\\
    &\EL : \FIB_0\,\U\qquad
    &&\ELCODE : \{\Gamma:\SET\}(\Phi : \FIB_0\,\Gamma) \FUN
    {\EL [\CODE\,\Phi] \EQ \Phi}
  \end{aligned}
\end{equation}
where for simplicity we restrict attention to fibrations whose fibers
are in the lowest universe, $\SET=\SET_0$. Here $\U$ is the
universe\footnote{Our predicative treatment of cofibrant types makes
  it necessary to place $\U$ in $\SET_2$ rather than $\SET_1$.} and
$\EL$ is a CCHM fibration over it which is a weak classifier in the
sense that any fibration $\Phi: \FIB_0\,\Gamma$ can be obtained from
it (up to equality) by re-indexing along some function
$\CODE\,\Phi:\Gamma\FUN\U$. (The word ``weak'' refers to the fact that
we do not require there to be a \emph{unique} function
$\gamma:\Gamma\FUN\U$ with $\EL [\gamma] \EQ \Phi$.)

We will show that the data in~\eqref{eq:9} implies that the interval
must be trivial ($\O\EQ\I$), contradicting the assumption in
\eqref{eq:3}.  This is because \eqref{eq:9} allows one to deduce that
if a family of types $A:\Gamma\FUN\SET$ has the property that each
$A\,x$ has a fibration structure when regarded as a family over the
unit type $\UNIT$, then there is a fibration structure for the whole
family $A$; and yet there are families where this cannot be the case.
For example, consider the family $P:\INT\FUN\SET$ with
$P\,i = (\O\EQ i)$.  For each $i:\INT$, the type $P\,i$ has a
fibration structure
$\pi\,i : \ISFIB\,\UNIT\,(\lambda\,\HOLE\FUN P\,i)$, because of
uniqueness of identity proofs~\eqref{eq:2}.  But the family as a whole
satisfies $\ISFIB\,\INT\,P \FUN\EMPTY$, because if we had a fibration
structure $\alpha: \ISFIB\,\INT\,P$, then we could apply it to
\begin{align*}
  &\ID :\PATH\,\INT
  &&\varphi : \SET
  &&u : \COF\,\varphi
  &&p : (i:\INT)\FUN\varphi \FUN P\,i
  &&z : \SIGMA{a_0}{P\,\O} p\,\O\EXT a_0\\
  &\ID\,i = i
  &&\varphi = \EMPTY
  &&u = \COFEMPTY
  &&p\,i = \lambda\,\HOLE\FUN \EMPTYELIM
  &&z = \PAIR{\REFL}{\lambda\,\HOLE\FUN\UIP}
\end{align*}
(where $\EMPTYELIM : \{A : \SET\}\FUN\EMPTY\FUN A$ is the elimination
function for the empty type) to get
$\alpha\,\ID\,\varphi\,u\,p\,z : (\SIGMA{a_1}{P\,\I} p\,\I\EXT a_1)$
and hence
$\INONTRIV\,(\FST\,(\alpha\,\ID\,\varphi\,u\,p\,z)): \EMPTY$.  From
this we deduce the following ``no-go''~\footnote{We are stealing
  Shulman's terminology~\cite[section~4.1]{ShulmanM:brofpt}.} theorem
for internal universes of CCHM fibrations.

\begin{theorem}{\normalfont\cite[Remark~7.5]{PittsAM:aximct-jv}}
  \label{thm:nogo}
  The existence of types and functions as in \eqref{eq:9} for CCHM
  fibrations is contradictory. More precisely, if $\INTUNIV :\SET_3$
  is the dependent record type with fields $\U$, $\EL$, $\CODE$ and
  $\ELCODE$ as in \eqref{eq:9}, then there is a term of type
  $\INTUNIV\FUN\EMPTY$.
\end{theorem}
\begin{proof}\hspace*{-4pt}\footnote{See the file
    \texttt{theorem-3-1.agda} at
    \url{https://doi.org/10.17863/CAM.22369} for an Agda version of
    this proof.}  Suppose we have an element of $\INTUNIV$ and hence
  functions as in \eqref{eq:9}. Then taking $P$ to be
  $\lambda i\FUN (\O\EQ i)$ and using the family $\pi\,i$ of fibration
  structures on each type $P\,i$ mentioned above, we get:
  \begin{equation}
    \label{eq:12}
    \Phi : \FIB_0\,\INT
    \qquad
    \Phi  =
    \EL[(\lambda\,i\FUN\CODE\,\PAIR{(\lambda\,\HOLE\FUN
      P\,i)}{\pi\,i}\,\UNITVAL)]
  \end{equation}
  Using $\ELCODE$ and function extensionality~\eqref{eq:1}, it follows
  that there is a proof $u : \FST\,\Phi \EQ P$, namely
  $u = \FUNEXT\,(\lambda\,i \FUN \CONG\,(\lambda\,x \FUN
  \FST\,x\,\UNITVAL)\, (\ELCODE\,\PAIR{(\lambda\,\HOLE\FUN
    P\,i)}{\pi\,i}))$, where $\CONG$ is the usual congruence property of
  equality.
  % ($\CONG\,f\,p: f\,x\EQ f\,y$, if $f:A \FUN B$, $x,y:A$ and
  % $p:x\EQ y$).
  From that and $\SND\,\Phi$ we get an element of
  $\ISFIB\,\INT\,P$. But we saw above how to transform such an element
  into a proof of $\EMPTY$. So altogether we have a proof of
  $\INTUNIV\FUN\EMPTY$. 
\end{proof}

\begin{rem}\normalfont 
  This counterexample generalizes to other notions of fibration: it is
  not usually the case that any type family $A : \Gamma \to \SET$ for
  which $A\,x$ is fibrant over $\UNIT$ for all $x : \Gamma$, is
  fibrant over $\Gamma$. The above proof should be compared with the
  proof that there is no ``fibrant replacement'' type-former in
  \emph{Homotopy Type System} (HTS); see
  \url{https://ncatlab.org/homotopytypetheory/show/Homotopy+Type+System#fibrant_replacement}.
  Theorem~\ref{thm:no-internal-tinyness} 
  below provides a further example of a global construct that does not
  internalize.
\end{rem}

%%%%%%%%%%%%%%%%%%%%%%%%%%%%%%%%%%%%%%%%%%%%%%%%%%%%%%%%%%%%%%%%%%%%%%
\section{Crisp type theory}
\label{sec:loctt}

The proof of Theorem~\ref{thm:nogo} depends upon the fact that in the
internal language the $\CODE$ function can be applied to elements with
free variables.  In this case it is the variable $i:\INT$ in
$\CODE\,\PAIR{(\lambda\,\HOLE\FUN P\,i)}{\pi\,i}\,\UNITVAL$; by
abstracting over it we get a function $\INT\FUN\U$ and re-indexing
$\EL$ along this function gives the offending fibration
\eqref{eq:12}. Nevertheless, the cubical sets presheaf topos does
contain a (univalent) universe which is a CCHM fibration classifier,
but only in an \emph{external} sense. Thus there is an object $\U$ in
$\Cset$ and a global section $\EL:1\morphism \FIB_0\,\U$ with the
property that for any object $\Gamma$ and morphism
$\Phi:1\morphism \FIB_0\,\Gamma$, there is a morphism
$\CODE\,\Phi: \Gamma\morphism \U$ so that $\Phi$ is equal to the
composition
${\FIB_0\,(\CODE\,\Phi)} \comp \EL : 1 \morphism \FIB_0\,\Gamma$;
see~\cite[Definition~18]{CoquandT:cubttc} for a concrete description
of $\U$.  The internalization of this property replaces the use of
global elements $1\morphism \Gamma$ of an object by local elements,
that is, morphisms $X\morphism \Gamma$ where $X$ ranges over a
suitable collection of generating objects (for example, the
representable objects in a presheaf topos); and we have
seen that such an internalized version cannot exist.

Nevertheless, we would like to explain the construction of universes
like $\U\in\Cset$ using some kind of type-theoretic language that
builds on Section~\ref{sec:intdcf}.  So we seek a way of manipulating
global elements of an object $\Gamma$, within the internal
language. One cannot do so simply by quantifying over elements of the
type $\UNIT\FUN\Gamma$, because of the isomorphism
$\Gamma \bij (\UNIT\FUN \Gamma)$.  Instead, we pass to a modal type
theory that can speak about global elements, which we call \emph{crisp
  type theory}. Its judgements, such as
$\Delta \mid \Gamma \vdash a : A$, have two context zones---where
$\Delta$ represents global elements and $\Gamma$ the usual, local
ones. The context structure is that used for an S4 necessitation
modality~\cite{PfenningF:judrml,DePaivaV:fibmtt,ShulmanM:brofpt},
because a global element from $\Delta$ can be used locally, but global
elements cannot depend on local variables from $\Gamma$.
Following~\cite{ShulmanM:brofpt}, we say that the left-hand context
$\Delta$ contains \emph{crisp} hypotheses about the types of
variables, written $x :: A$.

The interpretation of crisp type theory in cubical sets makes use of
the comonad $\flat:\Cset\morphism \Cset$ that sends a presheaf $A$ to
the constant presheaf on the set of global sections of $A$; thus
$\flat A(X) \cong A(1)$ for all $X\in\Cubes$ (where $1\in\Cubes$ is
terminal). Then a judgement $\Delta \mid \Gamma \vdash a : A$
describes the situation where $\Delta$ is a presheaf, $\Gamma$ is a
family of presheaves over $\flat\Delta$, $A$ is a family over
$\Sigma(\flat\Delta)\,\Gamma$ and $a$ is an element of that
family. The rules of crisp type theory are designed to be sound
for this interpretation. Compared with ordinary type theory, the key
constraint is that \emph{types in the crisp context and terms
  substituted for crisp variables depend only on crisp variables}. The
crisp variable and (admissible) substitution rules are:
\begin{equation}
  \label{eq:10}
  \inferrule{\ }{\Delta,x::A,\Delta' \mid \Gamma\vdash x:A}
  \qquad
  \inferrule{\Delta\mid \Ec\vdash a:A \\
    \Delta, x :: A, \Delta' \mid \Gamma \vdash  b : B}
  {\Delta, \Delta'[a/x] \mid  \Gamma[a/x] \vdash b[a/x] : B[a/x]}
\end{equation}
The semantics of the variable rule, which says that global elements
can be used locally, uses the counit
$\varepsilon\,A:\flat A\morphism A$ of the comonad $\flat$ mentioned
above.  In the substitution rule, $\Ec$ stands for the empty list, so
$a$ and $A$ may only depend upon the crisp variables from $\Delta$.
The other rules of crisp type theory (those for $\Pi$ types, $\Sigma$
types, etc.) carry the crisp context along.  For our application we
do not need a type-former for $\flat$, but instead make use of crisp
$\Pi$ types (see, e.g.~\cite{DePaivaV:fibmtt,PientkaB:conmtt}), that
is, $\Pi$ types whose domain is crisp
\begin{equation}
  \label{eq:11}
  \inferrule{\Delta\mid\Ec\vdash A:\SET_m\\\\
             \Delta,x::A\mid\Gamma\vdash B:\SET_n}
            {\Delta\mid\Gamma\vdash (x::A)\FUN B : \SET_{m\sqcup n}}
  \quad
  \inferrule{\Delta,x::A\mid \Gamma \vdash b : B}
            {\Delta\mid \Gamma \vdash \lambda x::A.b : (x::A)\FUN B }
  \quad
  \inferrule{\Delta \mid \Gamma \vdash f : (x::A)\FUN B \\\\
             \Delta \mid \Ec \vdash a: A}
            {\Delta \mid \Gamma \vdash f\,a : B[a/x]}  
\end{equation}
with $\beta\eta$ judgemental equalities.  In these rules, because the
argument variable $x$ is crisp, its type $A$, and the term $a$ to
which the function $f$ is applied, must also be crisp.
We also use \emph{crisp induction} for identity
types~\cite{ShulmanM:brofpt}---identity
elimination with a family $y::A,p::x\EQ y \vdash C(y,p)$ whose
parameters are crisp variables, which is given by a term of type
\begin{equation}
  \label{eq:13}
  \{A :: \SET_n\}\{x :: A\}(C : (y :: A) (p :: x\EQ y) \FUN\SET_n)(z :
  C\,x\,\REFL) (y :: A) (p :: x\EQ y) \FUN C\,y\, p
\end{equation}
together with a $\beta$ judgemental equality.

\begin{rem}[\textbf{Presheaf models of crisp type theory}]\normalfont
  \label{rem:local-topos}
  Crisp type theory is motivated by the specific presheaf topos
  $\Cset$ from Definition~\ref{def:cset}. However, it seems that very
  little is required of a category $\cat$ for the presheaf topos
  $\widehat{\cat}$ to soundly interpret it using the comonad
  $\flat=p^*\comp p_*$, where $p_*$ takes the global sections of a
  presheaf and its left adjoint $p^*$ sends sets to constant
  presheaves.  This $\flat$ preserves finite limits (because it is the
  composition of functors with left adjoints---$p^*$ is isomorphic to
  the functor given by precomposition with $\cat\morphism 1$ and hence
  has a left adjoint given by left Kan extension along
  $\cat\morphism 1$).  Although the details remain to be worked out,
  it appears that to model crisp type theory with crisp $\Pi$ types
  and crisp identification induction (and moreover a $\flat$ modality
  with crisp $\flat$ induction, which we do not use here), the only
  additional condition needed is that this comonad is idempotent
  (meaning that the comultiplication
  $\delta:\flat\morphism \flat\comp\flat$ is an isomorphism).  This
  idempotence holds iff $\widehat{\cat}$ is a connected topos, which
  is the case iff $\cat$ is a connected category---for example, when
  $\cat$ has a terminal object.  If it does have a terminal object,
  then $\widehat{\cat}$ is a local
  topos~\cite[Sect.~C3.6]{JohnstonePT:skeett} and $\flat$ has a right
  adjoint; in which case,
  conjecturally~\cite[Remark~7.5]{ShulmanM:brofpt}, one gets a model
  of the whole of Shulman's spatial type theory, of which crisp type
  theory is a part.  In fact $\Cubes$ does not just have a terminal
  object, it has all finite products (as does any Lawvere theory) and
  from this it follows that $\Cset$ is not just local, but also
  cohesive~\cite{LawvereFW:axic}.
\end{rem}

\begin{rem}[\textbf{Agda-flat}]\normalfont
  \label{rem:agda-flat}
  Vezzosi has created a fork of Agda, called
  \emph{Agda-flat}~\cite{Agda-flat}, which allows us to explore crisp
  type theory. It adds the ability to use crisp variables\footnote{The
    Agda-flat concrete syntax for ``$x::A$'' is
    ``$x \mathrel{{:}\{\flat\}} A$''.} $x :: A$ in places where
  ordinary variables $x:A$ may occur in Agda, and checks the modal
  restrictions in the above rules.  For example, Agda-flat quite
  correctly rejects the following attempted application of a
  crisp-$\Pi$ function to an ordinary argument
  \[
    \mathsf{wrong} : (A :: \SET_n)(B : \SET_m) (f : (\_ :: A) \FUN B)
    (x : A) \FUN B \qquad \mathsf{wrong}\,A\,B\,f\,x = f(x)
  \]
  while the variant with $x::A$ succeeds.  This is a simple example of
  keeping to the modal discipline that crisp type theory imposes; for
  more complicated cases, such as occur in the proof of
  Theorem~\ref{thm:universe} below, we have found Agda-flat
  indispensable for avoiding errors. However, Agda-flat implements a
  superset of crisp type theory and more work is needed to understand
  their precise relationship. For example, Agda's ability to define
  inductive types leads to new types in Agda-flat, such as the $\flat$
  modality itself; and its pattern-matching facilities allow one to
  prove properties of $\flat$ that go beyond crisp type theory. Agda
  allows one to switch off pattern-matching in a module; to be safe we
  do that as far as possible in our development. Installation
  instructions for Agda-flat can be found at
  \url{https://doi.org/10.17863/CAM.22369}.
\end{rem}

%%%%%%%%%%%%%%%%%%%%%%%%%%%%%%%%%%%%%%%%%%%%%%%%%%%%%%%%%%%%%%%%%%%%%%
\section{Universes from tiny intervals}
\label{sec:tini}

In crisp type theory, to avoid the inconsistency in the ``no-go''
Theorem~\ref{thm:nogo}, we can weaken the definition of a universe in
\eqref{eq:9} by taking $\CODE$ and $\ELCODE$ to be crisp functions of
fibrations $\Phi$ (and implicitly, of the base type $\Gamma$ of the
fibration). For if $\CODE$ has type
$\{\Gamma :: \SET\}(\Phi :: \FIB_0\,\Gamma)(x : \Gamma) \FUN \U$, then
the proof of a contradiction is blocked when in \eqref{eq:12} we try
to apply $\CODE$ to $\Phi =\PAIR{(\lambda\,\HOLE\FUN P\,i)}{\pi\,i}$,
which depends upon the local variable $i:\INT$. Indeed we show in this
section that given an extra assumption about the interval type $\INT$
that holds for cubical sets, it is possible to define a universe with
such crisp coding functions which moreover are unique, so that one
gets a classifying fibration, rather than just a weakly classifying
one.

Recall from Definition~\ref{def:cset} that in the cubical sets model,
the type $\INT$ denotes the representable presheaf $\y\Int\in\Cset$ on
the object $\Int\in\Cubes$. Since $\Cubes$ has finite products, there
is a functor $\HOLE\times\Int:\Cubes\morphism\Cubes$. Pre-composition
with this functor induces an endofunctor on presheaves
$(\HOLE\times\Int)^*:\Cset\morphism\Cset$ which has left and right
adjoints, given by left and right Kan
extension~\cite[Chap.~X]{MacLaneS:catwm} along
$\HOLE\times\Int$. Hence by the Yoneda~Lemma, for any $F\in\Cset$ and
$X\in\Cubes$
\[
(\INT\FUN F)\,X \bij \Cset(\y X, \INT\FUN F) \bij \Cset(\y X
\times \y\Int, F) \bij \Cset(\y(X\times\Int),F) =
((\HOLE\times\Int)^*F)\,X 
\]
naturally in both $X$ and $F$. It follows that the exponential functor
$\PATH = \INT\FUN\HOLE:\Cset\morphism\Cset$ is naturally isomorphic to
$(\HOLE\times\Int)^*$ and hence not only has a left adjoint
(corresponding to product with $\INT$) but also a right adjoint.  The
significance of objects in a category with finite products that are
not only exponentiable (product with them has a right adjoint), but
also whose exponential functor has a right adjoint was first pointed
out by Lawvere in the context of synthetic differential
geometry~\cite{LawvereFW:towdst}. He called such objects ``atomic'',
but we will follow later usage~\cite{YetterD:rigaef} and call them
\emph{tiny}.\footnote{Warning: the adjective ``tiny'' is sometimes
  used to describe an object $X$ of a $\mathcal{V}$-enriched
  cocomplete category $\mathcal{C}$ for which the hom
  $\mathcal{V}$-functor
  $\mathcal{C}(X,\HOLE):\mathcal{C}\morphism\mathcal{V}$ preserves
  colimits; see~\cite{SattlerC:equepm} for example. We prefer Kelly's
  term \emph{small-projective object} for this property. In the
  special case that $\mathcal{V}=\mathcal{C}$ and $\mathcal{C}$ is
  cartesian closed and has sufficient properties for there to be an
  adjoint functor theorem, then a small-projective object is in
  particular a tiny one in the sense we use here.}  Thus the interval
in cubical sets is tiny and we have a right adjoint to the path
functor $\PATH$ that we denote by $\COP:\Cset\morphism\Cset$. So for
each $B\in\Cset$, the functor
$\Cset(\PATH\,\HOLE,B) : \Cset\morphism\Set$ is representable by
$\COP B$, that is, there are bijections
$\Cset(\PATH\,A,B) \bij \Cset(A,\COP B)$, natural in $A$.

Given $\Gamma$ and $A:\Gamma\FUN\SET$ in $\Cset$, from
Definition~\ref{def:cchmf} we have that fibration structures
$1\morphism \ISFIB\,\Gamma\,A$ correspond to sections of
$\FST:(\SIGMA{p}{\PATH\,\Gamma} \C(A\comp p))\morphism \PATH\,\Gamma$
and hence, transposing across the adjunction $\PATH\dashv\COP$, to
morphisms making the outer square commute in the right-hand diagram
below:
\[
  \xymatrix{{\PATH\,\Gamma}\ar@{..>}[r] \ar[dr]_{\ID} &
    {\SIGMA{p}{\PATH\,\Gamma}\C(A\comp p)} \ar[d]^{\FST}\\
    &{\PATH\,\Gamma}}
  \;\;
  \xymatrix{
    \Gamma \ar@{..>}@/^1.2pc/[rrr] \ar@{..>}[r] \ar[dr]_{\ID} &
    {R_\Gamma A} \pullbackcorner \ar[rr]^<<<<<<<<{\pi_2} \ar[d]^{\pi_1} &&
    {\COP(\SIGMA{p}{\PATH\,\Gamma}\C(A\comp p))}
    \ar[d]^{\COP\FST}\\
    & \Gamma \ar[rr]_<<<<<<<<<<{\eta_\Gamma} &&
    {\COP(\PATH\,\Gamma)}}
\]
We therefore have that fibration structures for $A$ correspond to
sections of the pullback $\pi_1:R_\Gamma A \morphism \Gamma$ of
$\COP\FST$ along the unit
$\eta_\Gamma:\Gamma\morphism \COP(\PATH\,\Gamma)$ of the adjunction at
$\Gamma$ (which is the adjoint transpose of
$\ID:\PATH\,\Gamma\morphism\PATH\,\Gamma$).  This characterization of
fibration structure does not depend on the particular definition of
$\C$, so should apply to many notions of fibration. We will show how
it leads to the construction of a universe $\U= R_{\SET} \ID$ and
family $\pi_1: R_{\SET} \ID \morphism \SET$ which is a classifier for
fibrations.  However, there are two problems that have to be solved in
order to carry out the construction within type theory:
\begin{itemize}
  
\item First, for $\ELCODE$ in \eqref{eq:9} to be an equality (rather
  than just an isomorphism), one needs the choice of $R_\Gamma A$ to
  be strictly functorial with respect to re-indexing along $\Gamma$
  (and hence to be a dependent right adjoint in the sense of
  \cite{CloustonR:moddtt}). 
  
\item Secondly, one cannot use ordinary type theory as the internal
  language to formulate the construction, because the right adjoint to
  $\PATH$ does not internalize, as the following theorem shows.
\end{itemize}

\begin{theorem}
  \label{thm:no-internal-tinyness}
  There is no \emph{internal} right adjoint to the path functor
  $\PATH:\Cset\morphism\Cset$ for cubical sets. In other words, there
  is no family of natural isomorphisms
  $(\PATH\,\HOLE\FUN B) \bij (\HOLE\FUN\COP B): \Cset\morphism\Cset$
  (for $B\in\Cset$).
\end{theorem}
\begin{proof}
  It is an elementary fact about adjoint functors that such a family
  of natural isomorphisms is also natural in $B$. Note that
  $\PATH\UNIT\bij \UNIT$. So if we had such a family, then we would
  also have isomorphisms
  $B \bij (\UNIT\FUN B) \bij (\PATH\,\UNIT\FUN B) \bij (\UNIT\FUN\COP
  B) \bij \COP B$ which are natural in $B$. Therefore $\COP$ would be
  isomorphic to the identity functor and hence so would be its left
  adjoint $\PATH$. Hence $\INT\FUN\HOLE$ and $\UNIT\FUN\HOLE$ would be
  isomorphic functors $\Cset\morphism \Cset$, which implies (by the
  internal Yoneda Lemma) that $\INT$ is isomorphic to the terminal
  object $\UNIT$, contradicting the fact that $\INT$ has two distinct
  global elements.
\end{proof}

\begin{figure}
    \centering
  \begin{align*}
    &\COP : (A :: \SET_n) \FUN \SET_n\\
    &\R : \{A :: \SET_n\}\{B ::
      \SET_m\}(f :: \PATH\,A\FUN B) \FUN A \FUN  \COP B\\
    &\L : \{A :: \SET_n\}\{B :: \SET_m\}(g :: A\FUN \COP B) \FUN
      \PATH\,A \FUN B\\ 
    &\L\R : \{A :: \SET_n\}\{B :: \SET_m\}\{f ::
      \PATH\,A \FUN B\} \FUN \L (\R\,f) \EQ f\\
    &\R\L : \{A :: \SET_n\}\{B :: \SET_m\}\{g ::
      A \FUN \COP B\} \FUN \R (\L\,g) \EQ g\\
    &\R\PATH : \{A :: \SET_n\}\{B ::
       \SET_m\}\{C :: \SET_k\}(g :: A \FUN B)(f :: \PATH\,B \FUN  C)
      \FUN  \R (f \comp \PATH\ACT g) \EQ \R f \comp g     
  \end{align*}
  \vspace{-0.25in}
  \caption{Axioms for tinyness of the interval in crisp type theory}
  \label{fig:axitil}
\end{figure}

We will solve the first of the two problems mentioned above in the
same way that Voevodsky~\cite{VoevodskyV:csyduc} solves a similar
strictness problem (see also \cite[Section~6]{CloustonR:moddtt}):
apply $\COP$ once and for all to the displayed universe and then
re-index, rather than \emph{vice versa} (as done above).  The second
problem is solved by using the crisp type theory of the previous
section to make the right adjoint $\COP$ suitably global. The axioms
we use are given in Fig.~\ref{fig:axitil}. The function $\R$ gives the
operation for transposing (global) morphisms across the adjunction
$\PATH\dashv\COP$, with inverse $\L$ (the bijection being given by
$\R\L$ and $\L\R$); and $\R\PATH$ is the naturality of this
operation. The other properties of an adjunction follow from these, in
particular its functorial action
$\COP\ACT : \{A :: \SET_n\}\{B :: \SET_m\}(f :: A \FUN B) \FUN \COP A
\FUN \COP B$.  Note that Fig.~\ref{fig:axitil} assumes that the right
adjoint to $\INT\FUN(\HOLE)$ preserves universe levels. The soundness
of this for $\Cset$ relies on the fact that this adjoint is given by
right Kan extension~\cite[Chap.~X]{MacLaneS:catwm} along
$\HOLE\times\Int:\Cubes \morphism\Cubes$ and hence sends a presheaf
valued in the $n$th Grothendieck universe to another such.

\begin{theorem}[\textbf{Universe construction\protect\footnote{We just
      construct a universe for fibrations with fibers in $\SET_0$;
      similar universes $\U_n:\SET_{2\sqcup n}$ can be constructed for
      fibrations with fibers in $\SET_n$, for each $n$; see
      \texttt{theorem-5-2.agda} at
      \url{https://doi.org/10.17863/CAM.22369}.}}]
  \label{thm:universe}
  For fibrations as in Definition~\ref{def:cchmf} with any definition
  of composition structure $\C$ (e.g.~the CCHM one in \eqref{eq:22}),
  assuming axioms \eqref{eq:1}--\eqref{eq:4} and a tiny
  (Fig.~\ref{fig:axitil}) interval, there is a universe $\U$ equipped
  with a fibration $\EL$ which is \emph{classifying} in the sense that
  we have
  \begin{equation}
    \label{eq:20}
    \begin{array}{l}
      \U : \SET_2 \\
      \EL : \FIB_0\,\U \\
      \CODE :
      \{\Gamma :: \SET\}(\Phi :: \FIB_0\,\Gamma) \FUN \Gamma \FUN \U\\
      \ELCODE : \{\Gamma :: \SET\}(\Phi :: \FIB_0\,\Gamma) \FUN
      {\EL[\CODE\,\Phi] \EQ \Phi}\\
      \CODEEL : \{\Gamma :: \SET\}(\gamma :: \Gamma\FUN\U) \FUN
      \CODE(\EL[\gamma]) \EQ \gamma
    \end{array}
  \end{equation}
\end{theorem}
\begin{proof} Consider the display function associated with the first
  universe:
  \begin{align}
    &\ELT_1 : \SET_2
    &&\PR_1 : \ELT_1 \FUN \SET_1 \label{eq:21}\\
    &\ELT_1 = \SIGMA{A}{\SET_1} A
    &&\PR_1 (A , x) = A \notag    
  \end{align}
  We have $\C : \PATH\,\SET_0 \FUN \SET_1$ and
  hence  using the transpose operation from Fig.~\ref{fig:axitil},
  $\R\,\C:\SET_0\FUN\COP\,\SET_1$. We define $\U:\SET_2$ by taking
  a pullback:
  \begin{equation}
    \label{eq:23}
      \xymatrix{\U \pullbackcorner \ar[r]^{\pi_2} \ar[d]_{\pi_1} &
        \COP\,\ELT_1 \ar[d]^{\COP\ACT\PR_1} \\
        \SET \ar[r]_{\R\,\C} &
        {\COP\,\SET_1}}
    \qquad
    \begin{array}[t]{l}
      \U = \SIGMA{A}{\SET}(\SIGMA{B}{\COP\,\ELT_1} (\COP\ACT\PR_1\,B
      \EQ \R\,\C\,A))\\
      \pi_1 \PAIR{A}{\PAIR{\HOLE}{\HOLE}} = A\\
      \pi_2 \PAIR{\HOLE}{\PAIR{B}{\HOLE}} = B
    \end{array}
  \end{equation}
  Transposing this square across the adjunction $\PATH\dashv\COP$
  gives
  $\PR_1\comp \L\,\pi_2 = \C\comp \PATH\ACT\pi_1 : \PATH\,\U \FUN
  \SET_1$.  Considering the first and second components of
  $\L\,\pi_2 $, we have
  $\L\,\pi_2 \EQ
  \langle{\C\comp\PATH\ACT\pi_1}\mathbin{,}\upsilon\rangle$ for
  some $\upsilon : (p:\PATH\,\U)\FUN \C(\PATH\ACT\pi_1\,p)$; hence
  $\upsilon$ is an element of $\ISFIB\,U\,\pi_1$ and so we can define
  \begin{equation}
    \label{eq:25}
    \EL: \FIB_0\,\U
    \qquad \EL = (\pi_1,\upsilon)
  \end{equation}
  So it just remains to construct the functions in
  \eqref{eq:20}. Given $\Gamma :: \SET$ and
  $\Phi = (A,\alpha) ::\FIB_0\,\Gamma$, we have
  $\alpha :: \ISFIB\,\Gamma\,A = (p :\PATH\,\Gamma)\FUN
  \C(A\comp p)$. So the outer square in the left-hand diagram
  below commutes:
  \begin{equation}
    \label{eq:26}
      \xymatrix{
      {\PATH\,\Gamma} \ar@/_1pc/[ddr]_{\PATH\ACT\,A}
      \ar@/^1pc/[rrd]^{\langle \C\comp\PATH\ACT A \mathbin{,}
        \alpha\rangle}
      \ar@{..>}[dr] |{\PATH\ACT(\CODE\,\Phi)}\\
      & {\PATH\,\U} \ar[r]^{\L\,\pi_2} \ar[d]^{\PATH\ACT\pi_1}&
      \ELT_1 \ar[d]^{\PR_1}\\
      & {\PATH\,\SET} \ar[r]_{\C}&
      \SET_1}
    \qquad\qquad
    \xymatrix{
      {\Gamma} \ar@/_1pc/[ddr]_{A}
      \ar@/^1pc/[rrd]^{\R\,\langle \C\comp\PATH\ACT A \mathbin{,}
        \alpha\rangle} \ar@{..>}[dr] |{\CODE\,\Phi}
      \\
      & {\U} \pullbackcorner \ar[r]^{\pi_2} \ar[d]^{\pi_1}&
      {\COP\,\ELT_1} \ar[d]^{\COP\ACT \PR_1}\\
      & {\SET} \ar[r]_{\R\,\C}&
      {\COP\,\SET_1}}
  \end{equation}
  Transposing across the adjunction $\PATH\dashv\COP$, this means that
  the outer square in the right-hand diagram also commutes and
  therefore induces a function $\CODE\,\Phi:\Gamma\FUN \U$ to the
  pullback. So there are proofs of $\pi_1\comp\CODE\,\Phi \EQ A$ and
  $\pi_2\comp \CODE\,\Phi \EQ \R\, \langle \C\comp\PATH\ACT A
  \mathbin{,} \alpha\rangle$. Transposing the latter back across the
  adjunction gives a proof of
  $\L\,\pi_2\comp \PATH\ACT(\CODE\,\Phi) \EQ
  \langle\C\comp\PATH\ACT A\mathbin{,} \alpha\rangle$; and since
  $\L\,\pi_2 \EQ
  \langle{\C\comp\PATH\ACT\pi_1}\mathbin{,}\upsilon\rangle$, this
  in turn gives a proof of
  $\upsilon\comp \PATH\ACT(\CODE\,\Phi) \EQ \alpha$. Combining this
  with the proof of $\pi_1\comp\CODE\,\Phi \EQ A$, we get the desired
  element $\ELCODE\,\Phi$ of
  $\EL[\CODE\,\Phi] =
  \PAIR{\pi_1\comp\CODE\,\Phi}{\upsilon\comp\CODE\,\Phi} \EQ
  \PAIR{A}{\alpha} = \Phi$. Finally, taking $\Gamma=\U$ and $\Phi=\EL$
  in \eqref{eq:26}, the uniqueness property of the pullback implies
  that $\CODE\,\EL \EQ \ID$; and similarly, for any
  $\gamma::\Delta\FUN\Gamma$ we have that
  $(\CODE\,\Phi)\comp\gamma \EQ \CODE(\Phi[\gamma])$. Together these
  properties give us the desired element $\CODEEL$ of
  $\CODE(\EL[\gamma]) \EQ (\CODE\,\EL)\comp\gamma \EQ \ID\comp\gamma =
  \gamma$. 
\end{proof}

\begin{rem}\normalfont
  The above theorem can be generalized by replacing the particular
  universe $\ID:\SET\FUN\SET$ by an arbitrary one
  $E_0:U_0\FUN\SET$. So long as the composition structure $\C$ lands
  in $U_0$, one can use the above method to construct a universe of
  fibrant types from among the $U_0$ types.\footnote{See
    \texttt{theorem-5-2-relative.agda} at
    \url{https://doi.org/10.17863/CAM.22369}.} The application of this
  generalization we have in mind is to directed type theory; for
  example one can first construct the universe of fibrant types in the
  CCHM sense and then make a universe of covariant discrete fibrations
  in the Riehl-Shulman~\cite{ShulmanM:typtsc} sense from the fibrant
  types (repeating the construction with a different interval object).
\end{rem}

\begin{rem}\normalfont
  The results in this section only make use of the fact that the
  functor $\COP:\Cset\morphism\Cset$ is right adjoint to the
  exponential $\INT\FUN(\_)$ and we saw at the beginning of this
  section why such a right adjoint exists. It is possible to give an
  explicit description of presheaves of the form $\COP\Gamma$, but so
  far we have not found such a description to be useful.
\end{rem}

%%%%%%%%%%%%%%%%%%%%%%%%%%%%%%%%%%%%%%%%%%%%%%%%%%%%%%%%%%%%%%%%%%%%%%
\section{Applications}
\label{sec:modhtt}
\begin{figure}
  \[
    \begin{array}{c}
      \begin{array}{lll}
        \HOLE\MIN\HOLE:\INT\FUN\INT\FUN\INT
        &\O{\MIN} : (i:\INT)\FUN \O\MIN i \EQ \O
        &{\MIN}\O : (i:\INT)\FUN i\MIN\O \EQ \O\\
        &\I{\MIN} : (i:\INT)\FUN \I\MIN i \EQ i
        &{\MIN}\I : (i:\INT)\FUN i \MIN\I \EQ i
        \\
        \REV:\INT\FUN\INT
        &\REV\REV: (i:\INT)\FUN \REV\,(\REV\,i) \EQ i
        &\REV\O : \REV\,\O \EQ \I
      \end{array}
      \\ 
      \begin{array}{l}
        \PROPCOF : (\varphi:\SET)(u\;v: \COF\,\varphi) \FUN u \EQ v
        \\
        \COFPROP : (\varphi:\SET)(\HOLE:\COF\,\varphi)(x\;y:\varphi) \FUN
        x\EQ y\\
        \COFEXT : (\varphi\;\psi:\SET)(\HOLE:\COF\,\varphi)
        (\HOLE:\COF\,\psi)(\HOLE:\varphi\FUN\psi)(\HOLE:\psi\FUN\varphi)
        \FUN {\varphi \EQ \psi}\\
         \COFO : (i:\INT) \FUN \COF\,(\O\EQ i) \\
         \COFI : (i:\INT) \FUN \COF\,(\I\EQ i)\\ 
        \COFOR : (\varphi\;\psi:\SET)(\HOLE:\COF\,\varphi)
        (\HOLE:\COF\,\psi) \FUN \COF\,(\varphi \OR \psi)\\
        \COFAND : (\varphi\;\psi:\SET)(\HOLE: \COF\,\varphi)(\HOLE:\varphi \FUN
        \COF\,\psi)
        \FUN \COF\,(\varphi \times \psi)\\
        \COF\forall\INT : (\varphi : \INT\FUN\SET)(\HOLE : (i :\INT)\FUN
        \COF\,(\varphi\,i)) \FUN \COF\,((i:\INT)\FUN \varphi\,i)\\
        \STRAX:
        \begin{array}[t]{@{}r}
          (\varphi{:}\SET)(\HOLE{:}\COF\,\varphi)(A:\SET_n) (t
          :\varphi\FUN(\SIGMA{B}{\SET_n}{A\bij B})) \FUN \hspace{2cm}\mbox{}\\
          \SIGMA{T}{(\SIGMA{B}{\SET_n}{A\bij B})}{t\EXT T}
        \end{array}
      \end{array}
    \end{array}
  \]
  \caption{Further axioms needed for the CCHM model}
  \label{fig:furact}
  \label{fig:furai}
\end{figure}

\subparagraph*{Models.} Theorem~\ref{thm:universe} 
is the missing piece that allows a completely internal
development of a model of univalent foundations based upon the CCHM
notion of fibration, albeit internal to crisp type theory rather than
ordinary type theory. One can define a CwF in crisp type theory whose
objects are crisp types $\Gamma::\SET_2$, whose morphisms are crisp
functions $\gamma::\Gamma'\FUN\Gamma$, whose families are crisp CCHM
fibrations $\Phi=(A,\alpha)::\FIB_0\,\Gamma$ and whose elements are
crisp dependent functions $f::(x:\Gamma)\FUN A\,x$. To see that this
gives a model of
univalent foundations one needs to prove:\\
(a) The CwF is a model of intensional type theory with $\Pi$-types and
inductive types ($\Sigma$-types, identity types,
booleans, $W$-types, \ldots).\\
(b) The type $\U::\SET_2$ constructed in Theorem~\ref{thm:universe} is
fibrant (as a family over the unit type).\\
(c) The classifying fibration $\Phi::\FIB_0\,\U$ satisfies the univalence
axiom in this CwF.
  
Although we have yet to complete the formal development in Agda-flat,
these should be provable from axioms \eqref{eq:1}--\eqref{eq:4} and
Fig.~\ref{fig:axitil}, together with some further assumptions about
the interval object and cofibrant types listed in
Fig.~\ref{fig:furact}.  Part (a) was carried out in prior work, albeit
in the setting with an impredicative universe of
propositions~\cite{PittsAM:aximct-jv}.  In the predicative version
considered here, we replace the impredicative universe of propositions
with axioms asserting that being cofibrant is a mere proposition
($\PROPCOF$), that cofibrant types are mere propositions ($\COFPROP$)
and satisfy propositional extensionality ($\COFEXT$). These axioms are
satisfied by $\Cset$ provided we interpret $\COF:\SET\FUN\SET$ as
$\COF\,A = \exists \varphi:\Omega\mathbin{,} {\varphi\in\mathsf{Cof}}
\wedge {A\EQ\{\HOLE:\UNIT\mid \varphi\}}$, using the subobject
$\mathsf{Cof}\mono\Omega$ corresponding to the face lattice in
\cite{CoquandT:cubttc} (see
~\cite[Definition~8.6]{PittsAM:aximct-jv}). Axioms $\COFO$, $\COFI$,
$\COFOR$, $\COFAND$, $\COF\forall\INT$ and $\STRAX$ correspond to the
axioms $\AX_5$--$\AX_9$ from~\cite{PittsAM:aximct-jv}; in $\STRAX$,
$\cong$ is the usual internal statement of isomorphism. $\COFAND$ is
the dominance axiom that guarantees that cofibrations compose. Note
that axiom $\COFOR$ uses an operation sending mere propositions
$\varphi$ and $\psi$ to the mere proposition $\varphi\OR\psi$ that is
the propositional truncation of their disjoint union; the existence of
this operation either has to be postulated, or one can add axioms for
\emph{quotient types}~\cite[Section~3.2.6.1]{HofmannM:extcit} to crisp
type theory, (of which propositional truncation is an instance), in
which case function extensionality~\eqref{eq:1} is no longer needed as
an axiom, since it is provable using quotient
types~\cite[Section~6.3]{HoTT}.  Since in this paper we have taken a
CCHM fibration to just give a composition operation for cofibrant
partial paths from $\O$ to $\I$ and not vice versa, in
Fig.~\ref{fig:furai} we have postulated a path-reversal operation
$\REV$; this and the other axioms for $\INT$ in that figure suffice to
give a ``connection algebra'' structure on $\INT$~\cite[axioms $\AX_3$
and $\AX_4$]{PittsAM:aximct-jv}.

Part (b) can be proved using a version of the \emph{glueing} operation
from \cite{CoquandT:cubttc}, which is definable within crisp type
theory as in \cite[Section~6]{PittsAM:aximct-jv} and
\cite[Section~4.3.2]{SpittersB:guactt}.  The strictness axiom $\STRAX$
in Fig.~\ref{fig:furact} is needed to define this; and the assumption
that cofibrant types are closed under $\INT$-indexed $\forall$
($\COF\forall\INT$) is used to define the
appropriate fibration structure for glueing.

Part (c) can be proved as in \cite[Section~6]{PittsAM:decua} using a
characterization of univalence somewhat simpler than the original
definition of Voevodsky~\cite[Section~2.10]{HoTT}. The axiom $\STRAX$
gets used to turn isomorphisms into paths; and the axiom
$\COF\forall\INT$ is used to ``realign'' fibration structures that
agree on their underlying types (see~\cite[Lemma~6.2]{PittsAM:decua}).

\begin{rem}[\textbf{The interval is connected}]\normalfont
  \label{rem:intc}
  Fig.~\ref{fig:furact} does not include an axiom asserting that the
  interval is connected, because that is implied by its tinyness
  (Fig.~\ref{fig:axitil}). Connectedness was postulated as $\AX_1$
  in~\cite{PittsAM:aximct-jv} and used to prove that CCHM fibrations
  are closed under inductive type formers (and in particular that the
  natural number object is fibrant). The
  proof~\cite[Thm~8.2]{PittsAM:aximct-jv} that the interval in cubical
  sets is connected essentially uses the fact that $\Cset$ is a
  cohesive topos (Remark~\ref{rem:local-topos}). However it also
  follows directly from the tinyness property: connectedness holds iff
  $(\INT\FUN\BOOL)\bij \BOOL$, where $\BOOL = \UNIT + \UNIT$ is the
  type of Booleans. Since we postulate that $\INT\FUN\HOLE$ has a
  right adjoint, it preserves this coproduct and hence
  $(\INT\FUN \BOOL) \bij (\INT\FUN\UNIT)+(\INT\FUN\UNIT) \bij
  \UNIT+\UNIT = \BOOL$.
\end{rem}

\begin{rem}[\textbf{Alternative models}]\normalfont
  \label{rem:omed}
  We have focussed on axioms satisfied by $\Cset$ and the CCHM notion
  of fibration in that presheaf topos. However, the universe
  construction in Theorem~\ref{thm:universe} also applies to the
  cartesian cubical set models~\cite{LicataDR:carctt}, and we expect it
  is possible to give proofs in crisp type theory of its fibrancy and
  univalence as well.

  In this paper we only consider ``cartesian'' path-based models of
  type theory, in which a path is an arbitrary function out of an
  interval object, or in other words, the path functor is given by an
  exponential. The models in \cite{LumsdainePL:simmuf} and
  \cite{CoquandT:modttc} are not cartesian in that sense---the path
  functors they use are right adjoint to certain functorial
  cylinders~\cite{GambinoN:frocrp} not given by cartesian
  product.\footnote{Furthermore, obvious candidates for an interval
    object are not necessarily tiny in those models---for example, for
    the 1-simplex $\Delta[1]$ the exponential $\Delta[1]\FUN(\_)$ in
    the topos $\hat{\Delta}$ of simplicial sets does not have a right
    adjoint; thanks to a referee for pointing this out.} However,
  those path functors do have right adjoints (given by right Kan
  extension to suitable ``shift'' functors on the domain category of
  the presheaf toposes involved) and universes in these models can be
  constructed using the method of Theorem~\ref{thm:universe}. (Our Agda
  proof of that theorem does not depend upon the path functor being an
  actual exponential.) A proof in crisp type theory that those
  universes are fibrant and univalent may require a modification of
  our axiomatic treatment of cofibrancy; we leave this for future
  work.
\end{rem}

\subparagraph*{Universe hierarchies.}  Given that there are many
notions of fibration that one may be interested in, it is natural to
ask how relationships between them induce relationships between
universes of fibrant types.  As motivating examples of this, we might
want a cubical type theory with a universe of fibrations with
\emph{regularity}, an extra strictness corresponding to the
computation rule for identity types in intensional type theory; or a
three-level directed type theory with non-fibrant, fibrant, and
co/contravariant universes.  Towards building such hierarchies, in the
companion code\footnote{see \texttt{proposition-6-2.agda} at
  \url{https://doi.org/10.17863/CAM.22369}}  we have shown in crisp
type theory that universes are functorial in the notion of fibration
they encapsulate: when one notion of fibrancy implies another, the
first universe includes the second.

\begin{propo}
Let $\C^1, \C^2 : \PATH\,\SET_n \FUN\SET_{1\sqcup n}$ be two
notions of composition, $\ISFIB^1$ and $\ISFIB^2$ the corresponding
fibration structures, and $\U^1$ and $\U^2$ the corresponding classifying
universes.  A morphism of fibration structures is a function
$f_{\Gamma,A} : \ISFIB^1\,\Gamma\,A \FUN \ISFIB^2\,\Gamma\,A$ for all
$\Gamma$ and $A$, such that $f$ is stable under reindexing (given $h :
\Delta \FUN \Gamma$, and $\phi : \ISFIB^1\,\Gamma\,A$,
$f_{\Gamma,A}(\phi) \circ (\PATH\ACT h) \equiv f_{\Delta,A \comp
  f}(\phi[h])$).  Then a morphism of fibrations $f$ induces a function
$\U^1 \FUN \U^2$, and this preserves identity and composition. \qed
\end{propo}

%%%%%%%%%%%%%%%%%%%%%%%%%%%%%%%%%%%%%%%%%%%%%%%%%%%%%%%%%%%%%%%%%%%%%%
\section{Conclusion}
\label{sec:con}

Since the appearance of the CCHM~\cite{CoquandT:cubttc} constructive
model of univalence, there has been a lot of work aimed at analysing
what makes this model tick, with a view to simplifying and
generalizing it. Some of that work, for example by Gambino and
Sattler~\cite{GambinoN:frocrp,SattlerC:equepm}, uses category theory
directly, and in particular techniques associated with the notion of
Quillen model structure. Here we have continued to pursue the approach
that uses a form of type theory as an internal language in which to
describe the constructions associated with this model of univalent
foundations~\cite{PittsAM:aximct-jv,SpittersB:guactt}. For
those familiar with the language of type theory, we believe this
provides an appealingly simple and accessible description of the
notion of fibration and its properties in the CCHM model and in
related models. We recalled why there can be no internal description
of the univalent universe itself if one uses ordinary type theory as
the internal language. Instead we extended ordinary type theory with a
suitable modality and then gave a universe construction that hinges
upon the tinyness property enjoyed by the interval in cubical sets. We
call this language \emph{crisp type theory} and our work inside it has
been carried out and checked using an experimental version of Agda
provided by Vezzosi~\cite{Agda-flat}.

%%%%%%%%%%%%%%%%%%%%%%%%%%%%%%%%%%%%%%%%%%%%%%%%%%%%%%%%%%%%%%%%%%%%%%
% \bibliographystyle{plainurl}
% \bibliography{intumh}

%%%%%%%%%%%%%%%%%%%%%%%%%%%%%%%%%%%%%%%%%%%%%%%%%%%%%%%%%%%%%%%%%%%%%%
\end{document}